\documentclass[a4paper,UKenglish,cleveref, autoref]{lipics-v2019}

\usepackage{units}

\nolinenumbers

\title{Sherali-Adams and the binary encoding of combinatorial principles} 

\titlerunning{Sherali-Adams and the binary encoding of combinatorial principles}

\author{Stefan Dantchev}{Department of Computer Science, Durham University, U.K.}{}{}{}
\author{Abdul Ghani}{Department of Computer Science, Durham University, U.K.}{}{}{}
\author{Barnaby Martin}{Department of Computer Science, Durham University, U.K.}{}{}{}

\authorrunning{S. Dantchev, A. Ghani and B. Martin}

\Copyright{S. Dantchev, A. Ghani and B. Martin}

\ccsdesc[500]{Theory of computation~Computational complexity and cryptography}
\ccsdesc[500]{Theory of computation~Proof complexity}

\keywords{
Propositional Proof Complexity, Lift-and-Project Methods, Binary encoding
}

\category{}

\relatedversion{}

\supplement{}

\acknowledgements{We thank Nicola Galesi for collaboration on the binary encoding in Resolution-type systems. In particular, some of our definitions come from joint work with him in \cite{DantchevGalesiMartin}.}

\EventEditors{John Q. Open and Joan R. Access}
\EventNoEds{2}
\EventLongTitle{42nd Conference on Very Important Topics (CVIT 2016)}
\EventShortTitle{CVIT 2016}
\EventAcronym{CVIT}
\EventYear{2016}
\EventDate{December 24--27, 2016}
\EventLocation{Little Whinging, United Kingdom}
\EventLogo{}
\SeriesVolume{42}
\ArticleNo{23}

\theoremstyle{plain}
\newtheorem{thm}{\protect\theoremname}
\theoremstyle{plain}
\newtheorem{lem}[thm]{\protect\lemmaname}
\theoremstyle{plain}
\newtheorem{fact}[thm]{\protect\factname}
\theoremstyle{plain}
\newtheorem{cor}[thm]{\protect\corollaryname}

\makeatother

\usepackage{multicol}
\usepackage{amsmath}
\providecommand{\corollaryname}{Corollary}
\providecommand{\factname}{Fact}
\providecommand{\lemmaname}{Lemma}
\providecommand{\theoremname}{Theorem}

\newcommand{\PHP}[0]{\ensuremath{\textsc{PHP}}}
\newcommand{\LNP}[0]{\ensuremath{\textsc{LNP}}}
\newcommand{\TLNP}[0]{\ensuremath{\textsc{TLNP}}}

\newcommand{\SA}[0]{\ensuremath{\textsc{SA}}}
\newcommand{\LS}[0]{\ensuremath{\textsc{LS}}}

\newcommand{\ppi}[0]{\ensuremath{\Pi^{n-1}_{n-2}}}

\DeclareMathOperator{\C}{\mathsf{C_n}}

\DeclareMathOperator{\UC}{\mathsf{Un-C_n}}

\DeclareMathOperator{\BC}{\mathsf{Bin-C_n}}

\newcommand{\tuple}[1]{\mathbf{#1}}
\begin{document}

\maketitle

\begin{abstract}
We consider the Sherali-Adams (\SA) refutation system together with the unusual \emph{binary} encoding of certain combinatorial principles. For the unary encoding of the Pigeonhole Principle and the Least Number Principle, it is known that linear rank is required for refutations in \SA, although both admit refutations of polynomial size. We prove that the binary encoding of the Pigeonhole Principle requires exponentially-sized \SA\ refutations, whereas the binary encoding of the Least Number Principle admits logarithmic rank, polynomially-sized \SA\ refutations. We continue by considering a refutation system between \SA\ and Lasserre (Sum-of-Squares). In this system, the Least Number Principle requires linear rank while the Pigeonhole Principle becomes constant rank.
\end{abstract}

\section{Introduction}
It is well-known that questions on the satisfiability of propositional CNF formulae may be reduced to questions on feasible solutions for certain Integer Linear Programs (ILPs). In light of this, several ILP-based proof (more accurately, refutation) systems have been suggested for propositional CNF formulae, based on proving that the relevant ILP has no solutions. Typically, this is accomplished by relaxing an ILP to a continuous Linear Program (LP), which itself may have (non-integral) solutions, and then modifying this LP iteratively until it has a solution iff the original ILP had a solution \textcolor{black}{(which happens at the point the LP has no solution)}. Among the most popular ILP-based refutation systems are Cutting Planes \cite{Gomory1960,Chvatal73} and several proposed by Lov\'asz and Schrijver \cite{LovaszS1991}.

Another method for solving ILPs was proposed by Sherali and Adams \cite{SheraliA90}, and was introduced as a propositional refutation system in \cite{StefanGap}. Since then it has been considered as a refutation system in the further works \cite{TCS2009,NarrowMaximallyLong}. The Sherali-Adams system (\SA) is of significant interest as a static variant of the Lov\'asz-Schrijver system without semidefinite cuts (\LS). It is proved in \cite{laurent01comparison} that the \SA\ rank of a polytope is less than or equal to its \LS\ rank; hence we may claim that \SA\  is at least as strong as \LS\ (though it is unclear whether it is strictly stronger).

Various fundamental combinatorial principles used in Proof Complexity may be given in first-order logic as sentences $\varphi$ with no finite models and in this article we will restrict attention to those in $\Pi_2$-form. Riis discusses in \cite{SorenGap} how to generate from \textcolor{black}{prenex} $\varphi$ a family of CNFs, the $n$th of which encodes that $\varphi$ has a model of size $n$, which are hence contradictions. Following Riis, it is typical to encode the existence of the witnesses \textcolor{black}{to an existentially quantified variable in longhand with a big disjunction, of the form $S_{\tuple{a},1} \vee \ldots \vee S_{\tuple{a},n}$, that we designate the \emph{unary encoding}. Here the arity of $\tuple{a}$ is the number of universally quantified variables preceding the existentially quantified variable, on which it might depend.}

As recently investigated in the  works \cite{DBLP:journals/siamcomp/FilmusLNRT15,DBLP:journals/jacm/BonacinaG15,DBLP:journals/siamcomp/BonacinaGT16,DBLP:journals/combinatorica/LauriaPRT17,DBLP:conf/focs/HrubesP17,DantchevGalesiMartin}, it may also be possible to encode the existence of such witnesses {\em succinctly} by the use of a \emph{binary encoding}. Essentially, the existence of the witness is now given implicitly as any propositional assignment to the relevant variables $S_{\tuple{a},1},\ldots,S_{\tuple{a},\log n}$, which we call $S$ for Skolem, gives a witness; whereas in the unary encoding a solitary true literal tells us which is the witness.
Combinatorial principles encoded in binary are interesting to study for Resolution-type systems since they still preserve the hardness of the combinatorial  principle while giving a more  succinct propositional representation. In certain cases this leads to obtain significant lower bounds   in an easier way than for the unary case \cite{DBLP:journals/siamcomp/FilmusLNRT15,DBLP:journals/siamcomp/BonacinaGT16,DBLP:journals/combinatorica/LauriaPRT17,DantchevGalesiMartin}. 

The binary encoding also implicitly enforces an at-most-one constraint at the same time as it does at-least-one. When some big disjunction $S_{\tuple{a},1}\vee \ldots \vee S_{\tuple{a},n}$ of the unary encoding is translated to constraints for an ILP it enforces $S_{\tuple{a},1} + \ldots + S_{\tuple{a},n} \geq 1$. Were we to insist that $S_{\tuple{a},1} + \ldots + S_{\tuple{a},n} = 1$ then we encode immediately also the at-most-one constraint. We paraphrase this variant as being (the unary) \emph{encoding with equalities} or ``\SA-with-equalities''.
 
The Pigeonhole Principle (\PHP), which essentially asserts that $n$ pigeons may not be assigned to $n-1$ holes such that no hole has more than one pigeon, and the Least Number Principle (\LNP), which asserts that a partially-ordered $n$-set possesses a minimal element, are ubiquitous in Proof Complexity. Typically (and henceforth) we work under the same name with their negations, which are expressible in ($\Pi_2$) first-order logic as formulae with no finite models.

In \cite{TCS2009} we have proved that the \SA\ rank of (the polytopes associated with) (the unary encoding of) each of the Pigeonhole Principle and Least Number Principles is $n-2$ (where $n$ is the number of pigeons and elements in the poset, respectively). It is known that \SA\ polynomially simulates Resolution (see \mbox{e.g.} \cite{TCS2009}) and it follows there is a polynomially-sized refutation in \SA\ of the Least Number Principle. That there is a polynomially-sized refutation in \SA\ of the Pigeonhole Principle is noted in  \cite{Rhodes07}.

In this paper we consider the binary encodings of the Pigeonhole Principle and the Least Number Principle as ILPs. We additionally consider their (unary) encoding with equalities. We first prove that the binary encoding of the Pigeonhole Principle requires exponential size in \SA. We then prove that the (unary) encoding of the Least Number Principle with equalities has \SA\ rank 2 \textcolor{black}{and polynomial size}. This allows us to prove that the binary encoding of the Least Number Principle has \textcolor{black}{\SA\ rank at most $2 \log n$ and polynomial size}.

The divergent behaviour of these two combinatorial principles is tantalising -- while the Least Number Principle becomes easier for \SA\ in the binary encoding (in terms of rank), the Pigeonhole Principle becomes harder (in terms of size). Such variable behaviour has been observed for the Pigeonhole Principle in Resolution, where the binary encoding makes it easier for treelike Resolution (in terms of size) \cite{DantchevGalesiMartin}. 

We continue by considering a refutation system \SA+Squares which is between \SA\ and Lasserre (Sum-of-Squares) \cite{Lasserre2001} (see also \cite{laurent01comparison} for comparison between these systems). \textcolor{black}{\SA+Squares appears as Static LS$_+$ in \cite{russians}.} In this system one can always assume the non-negativity of (the linearisation of) any squared polynomial. In contrast to our system \SA-with-equalities, we \textcolor{black}{see} that the rank of the unary encoding of the Pigeonhole Principle is 2, while the rank of the Least Number Principle is linear. We \textcolor{black}{prove this by showing} a certain moment matrix in positive semidefinite. Our rank results for the unary encoding can be contrasted in Table \ref{tab:wavy}.  

\subsection{Related Work}

In another paper \cite{OtherPaper}, the present authors show a lower bound for \LNP\ in Lasserre that is a natural companion to the lower bound given in Corollary \ref{cor:9}. The lower bound in Lasserre is $\Omega(\sqrt{n})$, which is weaker than the linear lower bound of Corollary \ref{cor:9}, while requiring a significantly more sophisticated proof. 


\begin{table}[h]
\begin{center}
\begin{tabular}{|c|c|c|c|}
\hline
unary case & \SA & \SA-with-equalities & \SA+Squares \\
\hline
\PHP & linear & linear & constant \\
\hline 
\LNP & linear & constant & linear \\
\hline
\end{tabular}
\hspace{0.5cm}
\begin{tabular}{|c|c|}
\hline
binary case & \SA \\
\hline
\PHP & exponential  \\
\hline 
\LNP & polynomial  \\
\hline
\end{tabular}
\end{center}

\begin{center}
\begin{tabular}{|c|c|c|c|}
\hline
unary case & \SA & \SA-with-equalities & \SA+Squares \\
\hline
\PHP & \cite{TCS2009} & Appendix (\cite{TCS2009}) & Theorem \ref{thm:10} (\cite{russians}) \\
\hline 
\LNP & \cite{TCS2009} & Theorem \ref{thm:8} & Theorem \ref{thm:11}\\
\hline
\end{tabular}
\hspace{0.5cm}
\begin{tabular}{|c|c|}
\hline
binary case & \SA \\
\hline
\PHP & Theorem \ref{thm:size-bound}  \\
\hline 
\LNP & Corollary \ref{cor:9}  \\
\hline
\end{tabular}
\end{center}
\caption{Rank based complexity for the unary encoding in different systems (on the left) and size based complexity for the binary encoding (on the right). The lower table shows where the corresponding result is proved.}
\label{tab:wavy}
\end{table}

\section{Preliminaries}
\label{sec:pre}

Let $[m]$ be the set $\{1,\ldots,m\}$.
Let us assume, without loss of much generality, that $n$ is a power of $2$. Cases where $n$ is not a power of $2$ are handled in the binary encoding by explicitly forbidding possibilities.

If $P$ is a propositional variable, then $P^0=\overline{P}$ indicates the negation of $P$, while $P^1$ indicates $P$.

From a CNF formula $F:= C_1 \wedge \ldots \wedge C_r$ in variables $v_1,\ldots,v_m$ we generate an ILP in $2m$ variables $Z_{v_\lambda}, Z_{\neg v_\lambda}$ ($\lambda \in [m]$). For literals $l_1,\ldots,l_t$ s.t. $(l_1 \vee \ldots \vee l_t)$ is a clause of $F$ we have the constraining inequality 
\[ (\ref{sec:pre}.1) \ \ \ Z_{l_1} + \ldots + Z_{l_t} \geq 1. \]
We also have, for each $\lambda \in [m]$, the equalities of negation
\[ (\ref{sec:pre}.2) \ \ \ Z_{v_\lambda} + Z_{\neg v_\lambda} = 1 \]
together with the bounding inequalities
\[ (\ref{sec:pre}.3) \ \ \ 0 \leq Z_{v_\lambda} \leq 1 \ \ \ \mbox{and} \ \ \ 0 \leq Z_{\neg v_\lambda} \leq 1. \]
Let $\mathcal{P}_0^F$ be the polytope specified by these constraints on the real numbers. It is clear that this polytope contains integral points iff the formula $F$ is satisfiable.

Sherali-Adams (\SA) provides a static refutation method that takes the polytope $\mathcal{P}^F_0$ defined by $(\ref{sec:pre}.1)-(\ref{sec:pre}.3)$ and \emph{$r$-lifts} it to another polytope $\mathcal{P}^F_r$ in $\sum_{\lambda=0}^{r+1} {2m\choose \lambda}$ dimensions.
Specifically, the variables involved in defining the polytope $\mathcal{P}^F_r$ are $Z_{l_1 \wedge \ldots \wedge l_{r+1}}$
($l_1,\ldots,l_{r+1}$ literals of $F$) and $Z_\emptyset$. 
Let us say that the term $Z_{l_1 \wedge \ldots \wedge l_{r+1}}$ has \emph{rank} $r$.
Note that we accept commutativity and idempotence of the $\wedge$-operator, e.g. $Z_{l_1 \wedge l_2}=Z_{l_2 \wedge l_1}$ and $Z_{l_1 \wedge l_1}=Z_{l_1}$. Also $\emptyset$ represents the empty conjunct (boolean true); hence we set $Z_\emptyset:=1$.
For literals $l_1,\ldots,l_t$, s.t. $(l_1 \vee \ldots \vee l_t)$ is a clause of $F$, we have the constraining inequalities
\[ (\ref{sec:pre}.1') \ \ \ Z_{l_1 \wedge D} + \ldots + Z_{l_t \wedge D} \geq Z_{D}, \]
for $D$ any conjunction of at most $r$ literals of $F$. We also have, for each $\lambda \in [m]$ and $D$ any conjunction of at most $r$ literals, the equalities of negation
\[ (\ref{sec:pre}.2') \ \ \ Z_{v_\lambda \wedge D} + Z_{\neg v_\lambda \wedge D} = Z_D \]
together with the bounding inequalities
\[ (\ref{sec:pre}.3') \ \ \ 0 \leq Z_{v_\lambda \wedge D} \leq Z_{D} \ \ \ \mbox{and} \ \ \ 0 \leq Z_{\neg v_\lambda \wedge D} \leq Z_{D}. \]
The $\SA$ \emph{rank} of the polytope $\mathcal{P}^F_0$ (formula $F$) is the minimal $i$ such that $\mathcal{P}^F_i$ is empty. Thus, the notation rank is overloaded in a consistent way, since $\mathcal{P}^F_i$ is specified by inequalities in variables of rank at most $i$. The largest $r$ for which $\mathcal{P}^F_r$ need be considered is $2m-1$, since beyond that there are no new literals to lift by. Even that is somewhat further than necessary, largely because, if the conjunction $D$ contains both a variable and its negation, it may be seen from the equalities of negation that $Z_D=0$. In fact, it follows from \cite{laurent01comparison} that the \SA\ rank of $\mathcal{P}^F_0$ is always $\leq m-1$ (for a contradiction $F$). Of course, in general, $\mathcal{P}^F_0$ is non-empty; in fact, if $F$ is a contradiction that does not admit refutation by unit clause propagation, this is the case (we may use unit clause propagation to assign $0-1$ values to some variables, thereafter assigning $1/2$ to those variables remaining). Note that it follows that any unsatisfiable Horn CNF $F$ (i.e., where each clause contains at most one positive variable) has \SA\ rank $0$, since $F$ must then admit refutation by unit clause propagation (which may be used to demonstrate $\mathcal{P}^F_0$ empty). 

The number of defining inequalities of the polytope $\mathcal{P}^F_r$ is exponential in $r$; hence a naive measure of \SA\ size would see it grow more than exponentially in rank. However, not all of the inequalities $(\ref{sec:pre}.1') - (\ref{sec:pre}.3')$ may be needed to specify the empty polytope. We therefore define the \SA\ \emph{size} of the polytope $\mathcal{P}^F_0$ (formula $F$) to be the size of a minimal subset of the inequalities $(\ref{sec:pre}.1') - (\ref{sec:pre}.3')$ of $\mathcal{P}^F_{2m}$ that specifies the empty polytope.

We note that, for $r' \leq r$, the defining inequalities of $\mathcal{P}^F_{r'}$ are consequent on those of $\mathcal{P}^F_{r}$. Equivalently, any solution to the inequalities of $\mathcal{P}^F_{r}$ gives rise to solutions of the inequalities of $\mathcal{P}^F_{r'}$, when projected on to its variables. If $D'$ is a conjunction of $r'$ literals, then $Z_{D \wedge D'} \leq Z_{D}$ follows by transitivity from $r'$ instances of $(\ref{sec:pre}.3')$. We refer to the property $Z_{D \wedge D'} \leq Z_{D}$ as \emph{monotonicity}. \textcolor{black}{Finally, let us note that $Z_{v \wedge \neg v}=0$ holds in $\mathcal{P}^F_1$ and follows from a single lift of an equality of negation.}

Let us now consider principles which are expressible as first-order formulae, with no finite models, in $\Pi_2$-form,  \mbox{i.e.}  as $\forall \vec x \exists \vec w \varphi(\vec x,\vec w)$ where $\varphi(\vec x,\vec y)$ is a formula built on a family of relations $\vec R$. 
For example  the \emph{Least Number Principle}, which states that a finite partial order has a minimal element is one of such principles. 
Its negation can be expressed in $\Pi_2$-form as:  
\[\forall x,y,z \exists w \ \neg R(x,x) \wedge (R(x,y) \wedge R(y,z) \rightarrow R(x,z)) \wedge R(x,w).\]
This can be translated into a unsatisfiable CNF using a unary encoding of the witness, as shown below alongside the binary encoding.

\noindent\begin{minipage}{.45\linewidth}
	\begin{gather*}
	\LNP_n: \mbox{\em\underline{Unary encoding}}\\
	\overline P_{i,i} \qquad \forall i \in [n] \\
	\overline P_{i,j} \vee \overline P_{j,k} \vee P_{i,k} \qquad \forall i,j,k \in [n] \\
	\overline{S}_{i,j} \vee P_{i,j}  \qquad \forall i,j \in [n]  \\
	\textstyle \bigvee_{i\in [n]} S_{i, j} \qquad \forall j \in [n] \\
	\end{gather*}
\end{minipage}%
\begin{minipage}{.45\linewidth}
	\begin{gather*}
	\LNP_n: \mbox{\em \underline{Binary encoding}} \\
	\overline P_{i,i} \qquad \forall x \in [n] \\
	\overline P_{i,j} \vee \overline P_{j,k} \vee P_{i,k} \qquad \forall i,j,k \in [n] \\
	\textstyle \bigvee_{i\in [\log n]} S^{1-a_i}_{i,j} \vee P_{j,a} \qquad \forall j, a \in [n]\\
	\mbox{where $a_1\ldots a_{\log n} = \mathrm{bin}(a)$} \\
	\end{gather*}
\end{minipage} \vskip\baselineskip 
\noindent Note that we placed the witness in the Skolem variables $S_{i,x}$ as the first argument and not the second, as we had in the introduction. This is to be consistent with the $P_{i,j}$ and the standard formulation of \LNP\ as the least, and not greatest, number principle.

Indeed, one can see how to generate a binary encoding of $\mathrm{C}$ from any  combinatorial principle $\mathrm{C}$ expressible as  
a first order formula in  $\Pi_2$-form with no finite models. 
Exact details can be found in Definition 4 in \cite{DantchevGalesiMartin} and are reproduced here in the appendix.

As a second example we consider  the \emph{Pigeonhole Principle} which  states  that a total mapping from $[m]$ to 
$[n]$ has necessarily a collision when $m$ and $n$ are integers with  $m>n$. The negation of its  relational form \textcolor{black}{for $n$ one less than $m$} can be expressed as a $\Pi_2$-formula as 
$$\forall x,y,z \exists w \ \neg R(x,0) \wedge (R(x,z) \wedge R(y,z) \rightarrow x=y) \wedge R(x,w)$$ 
\textcolor{black}{where $0$ represents the difference between $m$ and $n$}. Its usual unary and  binary propositional  encoding are:

\noindent\begin{minipage}{.45\linewidth}
	\begin{gather*}
	\PHP^m_n: \mbox{\em\underline{Unary encoding}} \\
	\textstyle \bigvee_{j=1}^{n}P_{i,j}  \qquad \forall i \in [m] \\
	\overline P_{i,j} \vee \overline P_{i',j} \qquad \forall i \not = i'\in [m], j\in [n]\\
	\end{gather*}
\end{minipage}%
\begin{minipage}{.45\linewidth}
	\begin{gather*}
	 \PHP^m_n: \mbox{\em \underline{Binary encoding}}\\
	 \textstyle \bigvee_{j=1}^{\log n} P_{i,j}^{1 - a_j}  \vee  \bigvee_{j=1}^{\log n}P_{i',j}^{1 - a_j}\\\forall a \in [n],  i \not = i' \in [m] \\
	 \mbox{where $a_1\ldots a_{\log n} = \mathrm{bin}(a)$} \\
	\end{gather*}
\end{minipage} \vskip\baselineskip
\noindent \textcolor{black}{where $0$ no longer appears now $m$ and $n$ are explicit}.
Properly, the Pigeonhole Principle should also admit $S$ variables (as with the \LNP) but one notices that the existential witness $w$ to the type \emph{pigeon} is of the distinct type \emph{hole}. Furthermore, pigeons only appear on the left-hand side of atoms $R(x,z)$ and holes only appear on the right-hand side. 
For the Least Number Principle instead, the transitivity axioms effectively enforce the type of $y$ appears on both the left- and right-hand side of atoms $R(x,z)$. This accounts for why, in the case of the Pigeonhole Principle, we did not need to introduce any new variables to give the binary encoding, yet for the Least Number Principle a new variable $S$ appears. However, our results would hold equally were we to have chosen the more complicated form of the Pigeonhole Principle. \textcolor{black}{Note that our formulation of the Least Number Principle is symmetric in the elements and our formulation of the Pigeonhole Principle is symmetric is each of the pigeons and holes.}

When we consider the Sherali-Adams $r$-lifts of, \mbox{e.g.}, the Least Number Principle, we will identify terms of the form $Z_{P_{i,j} \wedge \overline{S}_{i',j'} \wedge \ldots}$ as $P_{i,j} \overline{S}_{i',j'} \ldots$. Thus, we take the subscript and use overline for negation and concatenation for conjunction. This prefigures the multilinear notation we will revert to in Section \ref{sec:SA+AS}, but one should view for now $P_{i,j} \overline{S}_{i',j'} \ldots$ as a single variable and not a multilinear monomial.

Finally, we wish to discuss the encoding of the Least Number Principle and Pigeonhole Principle as ILPs \emph{with equality}. For this, we take the unary encoding but instead of translating the wide clauses (\mbox{e.g.} from the \LNP) from  $\bigvee_{i\in [n]} S_{i,x}$ to $S_{1,x} + \ldots + S_{n,x} \geq 1$, we instead use $S_{1,x} + \ldots + S_{n,x} = 1$. This makes the constraint at-least-one into exactly-one (which is a priori enforced in the binary encoding). A reader who does not wish to address the appendix should consider the Least Number Principle as the combinatorial principle of the following lemma.
\begin{lem} \label{lem:binconv}
Let $\mathrm{C}$ be any combinatorial principle expressible as a first order formula in  $\Pi_2$-form with no finite models.
Suppose the unary encoding of $\mathrm{C}$ with equalities has an \SA\ rank $r$ and size $s$. Then the binary encoding of $\mathrm{C}$ has an \SA\ rank at most $r \log n$ and size at most $s$.
\end{lem}
\begin{proof}
We take the \SA\ refutation of the unary encoding of $\mathrm{C}$ with equalities of rank $r$, in the form of a set of inequalities, and build an \SA\ refutation of the binary encoding of $\mathrm{C}$ of rank $r \log n$, by substituting terms $S_{x,a}$ in the former with  $S^{a_1}_{x,1}\ldots S^{a_{\log n}}_{x,\log n}$, where $a_1\ldots a_{\log n} = \mathrm{bin}(a)$, in the latter. Note that the equalities of the form
\[\sum_{a_1 \ldots a_{\log n}=\mathrm{bin}(a)} S^{a_1}_{x,1}\ldots S^{a_{\log n}}_{x,\log n} = 1\] follow from the inequalities (\ref{sec:pre}.2') and (\ref{sec:pre}.3'). Further, inequalities of the form $S^{a_1}_{x,1}\ldots S^{a_{\log n}}_{x,\log n} \leq P_{x,a}$ follow since $S_{x,j}\overline{S}_{x,j}=0$ for each $j \in [\log n]$.
\end{proof}

\section{The lower bound for the binary Pigeonhole Principle}


In this section we study the inequalities derived from the binary encoding the Pigeonhole principle.
We first prove a \textcolor{black}{certain \SA\ rank} lower bound for a version of
the binary \PHP, in which only a subset of the holes is available.
\begin{lem}
\label{lem:degree-bound}Let $H\subseteq\left[n\right]$ be a subset
of the holes and let us consider  binary $\PHP_{|H|}^{m}$ where each
pigeon can go to a hole in $H$ only. Any \SA\ refutation of  binary $\PHP_{|H|}^{m}$
\textcolor{black}{involves a term that mentions at least $\left|H\right|$ pigeons}.
\end{lem}

\begin{proof}
We get a valuation from a partial matching in an obvious way. We say
that a term $T=\text{\ensuremath{\prod P_{i_{j},k_{j}}^{b_{j}}}}$
mentions the set of pigeons $M = \left\{ i_{j}\right\} $. Let us denote
the number of available holes by $n' := |H|$. Every term that mentions at
most $n'$ pigeons is assigned a value $v\left(T\right)$ as follows.
The set of pigeons mentioned in $M$ is first extended arbitrarily to a
set $M'$ of exactly $n'$ pigeons. $v\left(T\right)$ is then the
probability that a matching between $M'$ and $H$ taken uniformly
at random is consistent with the term $T$. In other words, $v\left(T\right)$
is the number of perfect matchings between $M'$ and $H$ that are
consistent with $T$, divided by the total, $(n')$!. Obviously, this
value does not depend on how $P$ is extended to $M'$. Also, it is
symmetric, i.e. if $\pi$ is a permutation of the pigeons, $v\left(\ensuremath{\prod P_{i_{j},k_{j}}^{b_{j}}}\right)=\ensuremath{v\left(\prod P_{\pi\left(i_{j}\right),k_{j}}^{b_{j}}\right)}$.

All lifts of \textcolor{black}{axioms of equality} $P_{j,k} + \overline{P}_{j,k} = 1$ are automatically satisfied since a
matching consistent with $T$ is consistent either with $TP_{j,k}^{b}$
or with $TP_{j,k}^{1-b}$ but not with both, and thus
\[
v\left(T\right)=v\left(TP_{j,k}^{b}\right)+v\left(TP_{j,k}^{1-b}\right).
\]
Regarding the lifts of the disequality of two pigeons in one hole, that is the inequalities coming from the only clauses in the binary encoding of \PHP, it
is enough to observe that it is consistent with any perfect matching,
i.e. at least one variable on the LHS is one under such a matching.
Thus, for a term $T$, any perfect matching consistent with $T$ will
also be consistent with $TP_{i,k}^{1-b_{k}}$ or with $TP_{\textcolor{black}{i},k}^{1-b_{k}}$
for some $k$.
\end{proof}
The proof of the size lower bound for the  binary $\text{PHP}_{n}^{n+1}$
then is by a standard random-restriction argument combined with the
\textcolor{black}{rank} lower bound above. Assume w.l.o.g that $n$ is a perfect power
of two. For the random restrictions $\mathcal{R}$, we consider the
pigeons one by one and with probability $\nicefrac{1}{4}$ we assign
the pigeon uniformly at random to one of the holes still available.
We first need to show that the restriction is "good" with high probability, \mbox{i.e.} neither too big nor too small.
The former is needed so that in the restricted version we have a good lower bound,
while the latter will be needed to show that a good restriction \textcolor{black}{coincides} well any reasonably big term.
A simple application of a Chernoff bound gives the following
\begin{fact}
\label{fact:good-restrictions}If $\mathcal{\left|R\right|}$ is the
number of pigeons (or holes) \textcolor{black}{assigned} by $\mathcal{R},$
\begin{enumerate}
\item the probability that $\mathcal{\left|R\right|}<\frac{n}{8}$ is at
most $e^{-\nicefrac{n}{32}}$, and
\item the probability that $\mathcal{\left|R\right|}>\frac{3n}{8}$ is at
most $e^{-\nicefrac{n}{48}}$.
\end{enumerate}
\end{fact}

So, from now on, we assume that $\frac{n}{8}\leq\left|\mathcal{R}\right|\leq\frac{3n}{8}$.
We first prove that a given wide term, i.e. a term that mentions a
constant fraction of the pigeons, survives the random restrictions
with exponentially small probability.
\begin{lem}
\label{lem:restriction-kill}Let $T$ be a term that mentions at least
$\frac{n}{2}$ pigeons. The probability that $T$ does not evaluate
to zero under the random restrictions is at most $\left(\frac{5}{6}\right)^{\nicefrac{n}{16}}$.
\end{lem}

\begin{proof}
An application of a Chernoff bound gives the probability that fewer
than $\frac{n}{16}$ pigeons mentioned by $T$ are \textcolor{black}{assigned} by $\mathcal{R}$
is at most $e^{-\nicefrac{n}{64}}$. For each of these pigeons the
probability that a single bit-variable in $T$ belonging to the pigeon
is set by $\mathcal{R}$ to zero is at least $\frac{1}{5}$. This
is because when $\mathcal{R}$ \textcolor{black}{sets} the pigeon, and thus the bit-variable,
there were at least $\frac{5n}{8}$ holes available, while at most
$\frac{n}{2}$ choices set the bit-variable to one. Thus $T$ survives
under $\mathcal{R}$ with probability at most
$e^{-\nicefrac{n}{64}}+\left(\frac{4}{5}\right)^{\nicefrac{n}{16}} < \left(\frac{5}{6}\right)^{\nicefrac{n}{16}}$
\end{proof}
Finally, we can prove that
\begin{thm}
\label{thm:size-bound}Any \SA\ refutation of the  binary $\PHP_{n}^{n+1}$
has to contain at least $\left(\frac{6}{5}\right)^{\nicefrac{n}{16}} - 1$
terms.
\end{thm}

\begin{proof}
Assume for a contradiction, that there is a smaller refutation. \textcolor{black}{Apply} the random restriction above to get a possibly smaller refutation
of the  binary $\text{PHP}_{\nicefrac{5n}{8}}^{5\nicefrac{n}{8}+1}$ w.h.p.
In particular, that refutation has fewer than $\left(\frac{6}{5}\right)^{\nicefrac{n}{16}}$
terms of rank at least $\frac{n}{2}$. Then by the union-bound, that is by adding together
the probabilities from Fact \ref{fact:good-restrictions} that the restriction is bad to 
the probability from Lemma \ref{lem:restriction-kill} multiplied by the number of terms, we get
a total smaller than one.
This implies that there is a specific restriction that is good and leaves no terms of rank at least $\frac{n}{2}$
in an \SA refutation of the  binary $\text{PHP}_{\nicefrac{5n}{8}}^{5\nicefrac{n}{8}+1}$.
However, this contradicts Lemma \ref{lem:degree-bound}.
\end{proof}
We now consider the so-called weak  binary PHP, $\text{PHP}_{n}^{m}$, where
$m$ is potentially much larger than $n$. The weak unary $\text{PHP}_{n}^{m}$ is interesting because it admits (significantly) subexponential-in-$n$ refutations in Resolution when $m$ is sufficiently large \cite{BussP97}. It follows that this size upper bound is mirrored in \SA. However, as proved in \cite{DantchevGalesiMartin}, the weak binary $\text{PHP}_{n}^{m}$ remains almost-exponential-in-$n$ for minimal refutations in Resolution. We will see here that the weak binary $\text{PHP}_{n}^{m}$ remains almost-exponential-in-$n$ for minimally sized refutations in \SA. In this weak binary case, the random restrictions
$\mathcal{R}$ above do not work, so we apply quite different restrictions
$\mathcal{R}'$ that are as follows: for each pigeon select independently
a single bit uniformly at random and set it to $0$ or $1$ with probability
of $\nicefrac{1}{2}$ each.

We can easily prove the following
\begin{lem}
\label{lem:restriction-kill-weak}A term \textcolor{black}{$T$} that mentions $n'$
pigeons does not evaluate to zero under $\mathcal{R}'$ with probability
at most $e^{-\nicefrac{n'}{2\log n}}$.
\end{lem}

\begin{proof}
For each pigeon mentioned, the probability that the bit-variable present
in $T$ is \textcolor{black}{set} by the random restriction is $\frac{1}{\log n}$, and
if so, the probability that the bit-variable evaluates to zero is
$\frac{1}{2}$. Since this happens independently for all $n'$ mentioned
pigeons, the probability that they all survive is at most $\left(1-\frac{1}{2\log n}\right)^{n'}$
\end{proof}
Now, we only need to prove that in the restricted version of the pigeon-hole
principle, there is always a big enough term.
\begin{lem}
\label{lem:degree-bound-weak}The probability that an \SA\ refutation
of the  binary $\PHP_{n}^{m}$, \textcolor{black}{for $m > n$}, after $\mathcal{R}'$ does not contain
a term \textcolor{black}{mentioning $\frac{n}{2\log n}$ pigeons} is at most $e^{-\nicefrac{n}{32\log^{2}n}}$.
\end{lem}

\begin{proof}
We first apply a Chernoff bound to deduce that for each bit position
$k$, \textcolor{black}{$1\leq k\leq (\log n)$} and a value $b$, $0$ or $1$, the probability
that there are fewer than $\frac{m}{4\log n}$ pigeons for which the
$k$th bit is set to $b$ is at most $e^{-\nicefrac{m}{16\log n}}$.
By the union bound, the probability that this holds for some position
$k$ and some value $b$ is at most $(2\log n) e^{-\nicefrac{m}{16\log n}}$.
Thus, with probability exponentially close to one $\mathcal{R}'$
leaves at least $\frac{m}{4\log n}$  pigeons of each type $\left(k,b\right)$,
\mbox{i.e.} the $k$th bit of the pigeon is set to $b$. \textcolor{black}{Recalling $m\geq n$}, we now pick a set
of pigeons $P$ that has \textcolor{black}{$(*)$} precisely \textcolor{black}{$\frac{n}{4\log n}$} pigeons of
each type (and thus is of size \textcolor{black}{$\nicefrac{n}{2}$}).

We evaluate any term $T$ that mentions at most \textcolor{black}{$\frac{n}{4\log n}$}
pigeons by first embedding  this set of pigeons into $P$, \textcolor{black}{which we can do due to property $(*)$}, and then giving it a value
as before. That is, by taking the probability that a perfect matching
between $P$ and \textcolor{black}{some chosen set of $\nicefrac{n}{2}$ holes consistent with the random restriction}, is consistent
with $T$.
\textcolor{black}{
To finish the proof, we need to show that such a set of $\nicefrac{n}{2}$ exists. This follows from the max--flow min--cut theorem which shows that a matching that contains every pigeon from $P$ into the set of holes exists.}
\end{proof}
We now proceed as in the proof of Theorem \ref{thm:size-bound} to
deduce that any SA refutation of the  binary $\text{PHP}_{n}^{m}$
must have size exponential in $n$.
\begin{cor}
Any \SA\ refutation of the  binary $\PHP_{n}^{m}$, $m>n$, has to
contain at least $e^{\nicefrac{n}{32\log^{2}n}}$ terms.
\end{cor}

\begin{proof}
Assume for a contradiction, that there is a refutation with fewer
terms of rank at most $\frac{n}{2\log n}$. By Lemma \ref{lem:restriction-kill-weak}
and a union-bound, there is a specific restriction that evaluates
all these terms to zero. However, this contradicts Lemma \ref{lem:degree-bound-weak}
.\end{proof}

	\section{The Least Number Principle with equality}
	Recall that the unary \emph{Least Number Principle ($\LNP_n$) with equality} has the following set of SA axioms:
	\begin{gather}
	\textit{self}: P_{i, i} = 0 \label{self} \quad \forall \; i \in n\\
	\textit{trans}: P_{i, k} - P_{i, j} - P_{j, k} + 1 \geq 0 \quad \forall \; i, j, k \in [n] \label{trans} \\
	\textit{impl}: P_{i, j} - S_{i, j} \geq 0 \quad \forall \; i, j \in [n]  \label{impl}\\
	\textit{lower}: \sum_{i \in [n]} S_{i, j} - 1 = 0 \quad \forall \; j \in [n] \label{lower}
	\end{gather}
	
\noindent	Strictly speaking Sherali-Adams is defined for inequalities only. An equality axiom $a = 0$ is simulated by the two inequalities $a \geq 0, -a \geq 0$, which we refer to as the 
	\emph{positive} and \emph{negative} instances of that axiom, respectively.  Also, note that we have used $P_{i,j}+\overline{P}_{i,j}=1$ to derive this formulation. We call 
		two terms \emph{isomorphic} if one term can be gotten from the other by relabelling the indices appearing in the subscripts \textcolor{black}{by a permutation}.
	
	\begin{thm}
		For $n$ large enough, the \SA\ rank of the $\LNP_n$ with equality is at most $2$ and \SA\ size at most polynomial in $n$.
\label{thm:8}
	\end{thm}
	
	\begin{proof}
	Note that if the polytope $\mathcal{P}^{\LNP_n}_2$ is nonempty there must exist a point where any isomorphic variables are given the same value. We can find  such a point by averaging an asymmetric valuation over all permutations of $[n]$.\\
	So suppose towards a contradiction there is such a symmetric point. \textcolor{black}{First note $P_{i,i}=S_{i,i}=0$ by \textit{self} and \textit{impl}.}
	We start by lifting the $j$th instance of \textit{lower} by $P_{i, j}$ to get
	\[
		S_{i, j}P_{i, j} + \sum_{k \neq i, j}  S_{k, j} P_{i, j} = P_{i, j}
	\]
	Equating (by symmetry) the terms $S_{k, j} P_{i, j}$ this is actually
	\[
	S_{i, j}P_{i, j} + (n-2)  S_{k, j} P_{i, j} = P_{i, j}.
	\]
	Lift this by $S_{k, j}$ to get
	\[
	S_{k, j} S_{i, j}P_{i, j} + (n-2)  S_{k, j} P_{i, j} = S_{k, j} P_{i, j}.
	\]
	We can delete the leftmost term by proving it must be $0$. Let us take an instance of \textit{lower} lifted by $S_{k, j}P_{i, j}$ for any $k \neq i,j$ along with an instance of monotonicity $S_{k, j}S_{m, j} P_{i, j} \geq 0$ for every $m \neq j,k$:
\textcolor{black}{
	\begin{gather}
	S_{k, j}P_{i, j} \left(1 - \sum_{m \neq j} S_{m, j}\right)  + \sum_{m \neq j,k,i} S_{k, j}S_{m, j} P_{i, j} \nonumber\\
	= - \sum_{m \neq k, j} S_{k, j}S_{m, j} P_{i, j} + \sum_{m \neq j,k,i} S_{k, j}S_{m, j} P_{i, j} \nonumber\\
	= - S_{k, j}S_{i, j} P_{i, j} \label{multwitnesses}. 
	\end{gather}
	The left hand side of this equation is greater than $0$ so we can deduce $S_{k, j}S_{i, j} P_{i, j}=0$.}
	
	This results in 
	\[
	(n-2) S_{k, j} P_{i, j} = S_{k, j} P_{i, j} \quad \text{\textcolor{black}{which is}} \quad S_{k, j} P_{i, j} = 0.
	\]
	We lift \textit{impl} by $S_{i, j}$ to obtain $S_{i, j} \leq S_{i, j}P_{i, j}$. Monotonicity gives us the
	opposite inequality and we can proceed as if we had the equality $S_{k, j} P_{k, j} = S_{k, j}$.\\
	So repeating the derivation of $S_{k, j} P_{i, j} = 0$ for every $i \neq k$ and then adding  $S_{k, j} P_{k, j} = S_{k, j}$ gets us $\sum_m S_{k, j} P_{m, j} = S_{k, j}$. Repeating this again for every $k$ and summing up gives
	\[
	 0 = \sum_{k,m} S_{k, j} P_{m, j} - \sum_k S_{k, j} = \sum_{k,m}  S_{k, j} P_{m, j} - 1
	\]
	with the last equality coming from the addition of the positive \textit{lower} instance $\sum_{k} S_{k, j} - 1 = 0$. Finally adding the lifted \textit{lower} instance $P_{m, j} - \sum_k S_{k, j} P_{m, j}\textcolor{black}{=0}$ for every $m$ gives
	\begin{equation} \label{ordersum}
	\sum_m P_{m, j} = 1.
	\end{equation}
	By lifting the \textit{trans} axiom $P_{i, k} - P_{i, j} - P_{j, k} + 1 \geq 0$ by $P_{j, k}$ we get
	\begin{equation} \label{zeros}
	P_{i, k}P_{j, k} - P_{i, j} P_{j, k}  \geq 0 
	\end{equation}
\textcolor{black}{
	Now, due to a manipulation similar to \Cref{multwitnesses} using \Cref{ordersum}
	\begin{gather}
	P_{k, j}P_{i, j} \left(1 - \sum_{m \neq j} P_{m, j}\right)  + \sum_{m \neq j,k,i} P_{k, j}P_{m, j} P_{i, j} \nonumber\\
	= - \sum_{m \neq k, j} P_{k, j}P_{m, j} P_{i, j} + \sum_{m \neq j,k,i} P_{k, j}P_{m, j} P_{i, j} \nonumber\\
	= - P_{k, j}P_{i, j} P_{i, j}\\
	= - P_{k, j}P_{i, j} \label{multwitnesses2}. 
	\end{gather}
	}

\noindent Thus,	$P_{i, k}P_{j, k}$ must be zero whenever $i \neq j$. Along with \Cref{zeros} we derive $P_{i, j} P_{j, k} = 0$.
	 \textcolor{black}{Noting $P_{i, j}P_{j, i}=0$ follows from \textit{trans} and \textit{self}, we lift} 
	\Cref{ordersum} by $P_{j, x}$ for some $x$ \textcolor{black}{to} get
	\[
	 P_{j, x}\sum_m P_{m, j} = \sum_{m \neq x, j} P_{m, j}P_{j, x} = P_{j, x}
	\]
	where we know the left hand side is zero (\Cref{zeros}). Thus we can derive $P_{i, j} = 0$ for any $i$ and $j$, resulting in a contradiction when combined with \Cref{ordersum}.
	\end{proof}

	\begin{cor}
		The binary encoding of $\LNP_n$ has \SA\ rank at most $2\log n$ and \SA\ size at most polynomial in $n$.
			\label{cor:9}
	\end{cor}
	\begin{proof}
		Immediate from \cref{lem:binconv}.
	\end{proof}

\section{\SA+Squares}
\label{sec:SA+AS}
In this section we consider a proof system, \SA+Squares, based on inequalities of \textcolor{black}{multilinear} polynomials. We now consider axioms as degree-1 polynomials in some set of variables and refutations as polynomials in those same variables. Then this system is gotten from \SA\ by allowing addition of (linearised) squares of polynomials. In terms of strength this system will be strictly stronger than \SA\ and at most as strong as Lasserre (also known as Sum-of-Squares), although we do not at this point see \textcolor{black}{an exponential} separation between \SA+Squares and Lasserre.
See \cite{Lasserre2001, laurent01comparison, SoS-survey} for more on the Lasserre proof system \textcolor{black}{and \cite{Lauria2017} for tight degree lower bound results}.

Consider the polynomial $S_{i,j} P_{i,j} - S_{i,j} P_{i,k}$. The square of this is 
\[
S_{i,j} P_{i,j}S_{i,j} P_{i,j} + S_{i,j} P_{i,k} S_{i,j} P_{i,k} - 2 S_{i,j} P_{i,j} S_{i,j} P_{i,k}.
\]
Using idempotence this linearises to $S_{i,j} P_{i,j} + S_{i,j} P_{i,k} - 2 S_{i,j} P_{i,j} P_{i,k}$. Thus we know that this last polynomial is non-negative for all $0/1$ settings of the variables.\\
A \emph{degree-$d$} \SA+Squares refutation of a set of linear inequalities (over terms) $q_1 \geq 0, \ldots, q_x \geq 0$ is an equation of the form
\begin{equation} \label{eq:SA+Sref}
-1 = \sum_{i = 1}^x p_i q_i + \sum_{i = 1}^y r_i^2
\end{equation}
where the $p_i$ are polynomials with nonnegative coefficients and the degree of the polynomials $p_i q_i, r_i^2$ is at most $d$. 
We want to underline that we now consider a term like $S_{i,j} P_{i,j} P_{i,k}$ as a product of its constituent variables. This is opposed to the preceding sections in which we viewed it as a single
variable $Z_{S_{i,j} P_{i,j} P_{i,k}}$.
The translation from the degree discussed here to \SA\ rank previously introduced may be paraphrased by ``$\mathrm{rank}=\mathrm{degree}-1$''.

We show that the unary \PHP\ becomes easy in this stronger
proof system while the \LNP\ remains hard. \textcolor{black}{The following appears as Example 2.1 in \cite{russians} but we reproduce its easy proof for completeness.}
\begin{thm}[\cite{russians}]
	The $\PHP^{n+1}_n$ has an $\SA+\mathrm{Squares}$ refutation of degree 2.
	\label{thm:10}
\end{thm}
\begin{proof}
Indeed, for a hole $j$ we square and then linearise the polynomial
\[
1-\sum_{i=1}^{m}P_{i,j}
\]
to get the inequality
\begin{equation}
1-\sum_{i=1}^{m}P_{i,j}+2\sum_{1\leq i<i'\leq m}P_{i,j}P_{i',j}\geq0.\label{eq:hole-squared}
\end{equation}
(We have used the linearisation, $P_{i,j}^{2}=P_{i,j}$.) On the other
hand, by lifting each axiom $\overline{P}_{i,j}+\overline{P}_{i',j}\geq1$
by $P_{i,j}P_{i',j}$ we deduce $0\geq P_{i,j}P_{i',j}$. Multiplying
by $2$ these inequalities for all $i,i'$, $0\leq i<i'\leq m$ and
adding them to (\ref{eq:hole-squared}) gives
\[
1-\sum_{i=1}^{m}P_{i,j}\geq0.
\]
By adding over all holes, we get
\[
n-\sum_{j=1}^{n}\sum_{i=1}^{m}P_{i,j}\geq0.
\]
On the other hand, by adding all pigeon axioms, we get
\[
\sum_{i=1}^{m}\sum_{j=1}^{n}P_{i,j}\geq m.
\]
From the last two inequalities, we get the desired contradiction,
$n - m \geq 0$.
\end{proof}

We give our lower bound by producing a linear function $v$ (which we will call a \emph{valuation}) from terms into $\mathbb{R}$ such that
\begin{enumerate}
 \item for each axiom $p \geq 0$ and every term $X$ with $deg(Xp) \leq d$ we have $v(X p) \geq 0$, and 
 \item we have $v(r^2) \geq 0$ whenever $deg(r^2) \leq d$.
\end{enumerate}
The existence of such a valuation clearly implies that a degree-$d$ \SA+Squares
refutation cannot exist, as it would result in a contradiction when applied to both sides of \cref{eq:SA+Sref}.\\
To verify that $v(r^2) \geq 0$ whenever $deg(r^2) \leq d$ we show that the so-called \emph{moment-matrix} $\mathcal{M}_v$ is positive semidefinite. The degree-$d$ moment matrix is
defined to be the symmetric square matrix whose rows and columns are
indexed by terms of size at most $d/2$ and each entry is the valuation of the product of the two terms indexing that entry.
Given any polynomial $\sigma$ of degree at most $d/2$ let $c$ be its coefficient vector. Then if $\mathcal{M}_v$ is positive semidefinite:
\[
v(\sigma^2) = \sum_{deg(T_1), deg(T_2) \leq d/2} {c}(T_1) {c}(T_2) v(T_1 T_2) = {c}^\top \mathcal{M}_v {c} \geq 0.
\]
(For more on this see e.g. \cite{Lasserre2001}, section 2.)\\
\textcolor{black}{
	Recall that the unary \emph{Least Number Principle ($\LNP_n$)} has the following set of SA axioms:
	\begin{gather}
	\textit{self}: P_{i, i} = 0 \label{self2} \quad \forall \; i \in n\\
	\textit{trans}: P_{i, k} - P_{i, j} - P_{j, k} + 1 \geq 0 \quad \forall \; i, j, k \in [n] \label{trans2} \\
	\textit{impl}: P_{i, j} - S_{i, j} \geq 0 \quad \forall \; i, j \in [n]  \label{impl2}\\
	\textit{lower}: \sum_{i \in [n]} S_{i, j} - 1 \geq 0 \quad \forall \; j \in [n] \label{lower2}
	\end{gather}
	}
\begin{thm}
	There is no $\SA+\mathrm{Squares}$ refutation of the $\LNP_n$ with degree at most $(n-3)/2$.
\label{thm:11}
\end{thm}

\begin{proof}
For each term $T$, let $v\left(T\right)$ be the probability
that $T$ is consistent with a permutation on the $n$ elements taken
uniformly at random or, in other words, the number of permutations
consistent with $T$ divided by $n!$. \textcolor{black}{Here we view $S_{x,y}$ as equal to $P_{x,y}$.} This valuation trivially satisfies
the \textcolor{black}{lifts of the \textit{self} and \textit{trans}} axioms as they are satisfied
by each permutation (linear order). \textcolor{black}{It satisfies the lifts of the \textit{impl} axioms by construction.} We now claim that the lifts of
\textcolor{black}{the \textit{lower}} Skolem axioms (those containing only $S$ variables) 
of degree up to $\frac{n-3}{2}$ are also satisfied
by $v\left(.\right)$. Indeed, let us consider the lifting by $T$
of the Skolem axiom for $x$
\begin{equation}
\sum_{y=1}^{n}T\textcolor{black}{S}_{x,y}\geq T.\label{eq:there is smaller}
\end{equation}
Since $T$ mentions at most $n-3$ elements, there must be at least
two $y_{1}\neq y_{2}$ that are different from all of them and from
$x.$ For any permutation that is consistent with $T$, the probability
that each of the $y_{1}$ and $y_{2}$ is smaller than $x$ is precisely
a half, and thus
\[
v\left(T\textcolor{black}{S}_{x,y_{1}}\right)+v\left(T\textcolor{black}{S}_{x,y_{2}}\right)=v\left(T\right).
\]
Therefore the valuation of the LHS of (\ref{eq:there is smaller})
is always greater than or equal to the valuation of $T$.

Finally, we need to show that the valuation is consistent with the non-negativity of (the
linearisation of) any squared polynomial. It is easy to see that the moment matrix for $v$ can be written as
\[
\frac{1}{n!}\sum_{\sigma}V_{\sigma}V_{\sigma}^{T}
\]
where the summation is over all permutations on $n$ elements and
for a permutation $\sigma$, $V_{\sigma}$ is its characteristic vector.
The characteristic vector of a permutation $\sigma$ is a Boolean
column vector indexed by terms and whose entries are $1$ or $0$
depending on whether the respective index term is consistent or not
with the permutation $\sigma$. Clearly the moment matrix is positive semidefinite being
a sum of (rank one) positive semidefinite matrices.
\end{proof}
\textcolor{black}{
An alternative formulation of the Least Number Principle asks that the order be total, and this is enforced with axioms \textit{anti-sym} of the form $P_{i,j} \vee P_{j,i}$, or $P_{i,j}+P_{j,i} \geq 1$, for $i\neq j \in [n]$. Let us call this alternative formulation \TLNP. Ideally, lower bounds should be proved for  \TLNP, because they are potentially stronger. Conversely, upper bounds are stronger when they are proved on the ordinary \LNP, without the total order. Looking into the last proof, one sees that the lifts of \textit{anti-sym} are satisfied as we derive our valuation exclusively from total orders. This is interesting because an upper bound in Lasserre of order $\sqrt{n}$ is known for $\TLNP_n$ \cite{potechin2018sum}. Thus, Theorem \ref{thm:11}, together with \cite{potechin2018sum}, shows a quadratic separation between \SA+Squares and Lasserre. The question of an exponential separation remains open.
}

\section{Conclusion}

Our result that the unary encoding of the Least Number Principle with equalities has \SA\ rank 2 contrasts strongly with the fact that the unary encoding of the Least Number Principle has \SA\ rank $n-2$ \cite{TCS2009}. Now we know the unary encoding of the Pigeonhole Principle has \SA\ rank $n-2$ also. This leaves one wondering about the unary encoding of the Pigeonhole Principle with equalities, which does appear in Figure \ref{tab:wavy}. In fact, the valuation of \cite{TCS2009} witnesses this still has \SA\ rank $n-2$ (and we give the argument in the appendix). That is, the Pigeonhole Principle does not drop complexity in the presence of equalities, whereas the Least Number Principle does.



\section*{Appendix}

\subsection*{Binary versus unary encodings in general}

Here we reproduce Section 6 from \cite{DantchevGalesiMartin}. Let $\C$ be some combinatorial principle  expressible  as a first-order $\Pi_2$-formula $F$ of the form
$\forall \vec x \exists \vec w \varphi(\vec x,\vec w)$ where $\varphi(\vec x,\vec w)$ is a quantifier-free formula built on a family of 
relations $\vec R$. Following Riis \cite{SorenGap} we restrict to the class of such formulae having no finite model.  

Let $\UC$  be the standard unary (see Riis in \cite{SorenGap}) CNF propositional encoding of $F$. 
For each set of first-order variables $\vec a:=\{x_1, \ldots, x_k\}$ of (first order) variables, we consider the  propositional variables $v_{x_{i_1},x_{i_2}, \dots ,x_{i_k}}$ (which we abbreviate as  $v_{\vec a}$)  whose semantics are to capture at once the value of variables in $\vec a$ if they appear in some relation in $\varphi$. For easiness of description we restrict to the case where $F$ is of the form $\forall \vec x \exists w \varphi(\vec x,w)$, \mbox{i.e.} ${\vec w}$ is a single variable $w$. Hence the propositional variables of $\UC$ are of the type $v_{\vec a}$ for $\vec a\subseteq \vec x$ (type 1 variables) and/or of the type $v_{\vec xw}$ for $w \in {\vec w}$ (type 2 variables) and which we denote by simply $v_{w}$, 
since each existential variable in $F$ depends always on all universal variables.
Notice that we consider the case of $F= \forall \vec x \exists w \varphi(\vec x,w)$, since the  generalisation to higher arity is clear as each witness $w \in {\vec w}$ may be treated individually.

\begin{definition}(Canonical form of $\BC$)
\label{def:binC}
Let $\C$ be a  combinatorial principle  expressible as a first-order formula $\forall \vec x \exists w \varphi(\vec x,w)$ with no finite models. Let $\UC$ be its unary propositional encoding. Let $2^{r-1}<n\leq 2^r \in \mathbb N$ ($r=\lceil \log n\rceil$). The binary encoding $\BC$ of $C$ is defined as follows:

The {\em \bf variables} of $\BC$ are defined from variables of $\UC$ as follows: 
\begin{enumerate}
\item For each variable of type 1 $v_{\vec a}$,  for $\vec a \subseteq \vec x$,  we use a variable $\nu_{\vec x}$,  for  $\vec a \subseteq \vec x$, and
\item  For each variable of type 2 $v_w$, we have $r$ variables $\omega_1,\ldots \omega_r$, where we use the convention that if 
$z_1\ldots z_{r}$ is the binary representation of $w$, then 
$$ 
\omega^{z_j}_{j}=\left\{
   		\begin{array}{ll}
            		\omega_{j} &  z_j=1 \\
            		\overline \omega_{j} & z_j=0
    		\end{array} \right. 
$$ 
so that $v_{w}$ can be represented using binary variables by the clause $(\omega^{1-z_1}_{1} \vee \ldots \vee \omega^{1-z_r}_{r})$
\end{enumerate}

The  {\em \bf clauses} of $\BC$ are defined form the clauses of $\UC$ as follows:
\begin{enumerate}
\item If $C \in \UC$ contains only variables of type 1, $v_{\vec b_1},\dots, v_{\vec b_k}$, hence $C$ is mapped as follows
$$
\begin{array}{lll}
C:= \bigvee_{j=1}^{k_1} v_{\vec b_j} \vee \bigvee_{j=1}^{k_2} \overline v_{\vec c_j} &\mapsto &\bigvee_{j=1}^{k_1} \nu_{\vec b_j} \vee \bigvee_{j=1}^{k_2} \overline \nu_{\vec c_j}
\end{array}
$$


\item If $C \in \UC$ contains type 1  and type 2 variables, it is mapped as follows:
$$
\begin{array}{lll}
C:= v_{w} \vee \bigvee_{j=1}^{k_1} v_{\vec c_j} \vee \bigvee_{l=1}^{k_2} \overline v_{\vec d_j} & \mapsto & \left( 
\bigvee_{i \in [r]} \omega^{1-z_i}_{i}\right ) \vee \bigvee_{j=1}^{k_1} \nu_{\vec c_j} \vee \bigvee_{l=1}^{k_2} \overline \nu_{\vec d_j}\\
C:= \overline v_{w} \vee \bigvee_{j=1}^{k_1} v_{\vec c_j} \vee \bigvee_{l=1}^{k_2} \overline v_{\vec d_j} & \mapsto & \left( 
\bigvee_{i \in [r]} \omega^{z_i}_{i}\right ) \vee \bigvee_{j=1}^{k_1} \nu_{\vec c_j} \vee \bigvee_{l=1}^{k_2} \overline \nu_{\vec d_j}\\
\end{array}
$$
where $\vec c_j,\vec d_l\subseteq \vec x$ and where $z_1,\ldots,z_r$ is the binary representation of $w$.

\item If $n\neq 2^r$, then, for each $n<a\leq 2^r$ we need clauses
\[ \omega_{1}^{1-a_1} \vee \ldots \vee \omega_{r}^{1-a_r} \]
where $a_1,\ldots,a_r$ is the binary representation of $a$.
\end{enumerate}
\end{definition}

\subsection*{Equality version of Proposition 11 from \cite{TCS2009}}

The notation for the following proposition will make sense only in light of Section 4.2 from \cite{TCS2009}.

\

\noindent \textbf{Proposition 11bis}.
The given valuation is valid for the $(n-3)$rd lifts of the equalities. That is, for all $i$, if $|\Phi| \leq n-3$ then
\[ P_{\Phi \wedge (i,1)} + P_{\Phi \wedge (i,2)} + \ldots + P_{\Phi \wedge (i,n)} = P_\Phi. \]

\begin{proof}
Suppose that $i \neq n$ and $\Phi$ contains no instances of $n$. Let $P'$ be the proportion of $\pi \in \ppi$, consistent with $\Phi$, that leave $i$ unmapped. It follows from our model counting that
\[ P_{\Phi \wedge (i,1)} + P_{\Phi \wedge (i,2)} + \ldots + P_{\Phi \wedge (i,n-1)} + P' = P_\Phi. \]
So, it suffices to prove that
\[ P_{\Phi \wedge (i,n)} = P'. \]
There must be some $j' \in [n-1] \setminus \{1\}$, s.t. $j'$ does not appear on the right-hand side of any atom in $\Phi$, whereupon  $P_{\Phi \wedge (i,n)}= P_{\Phi \wedge (i,j')}$ and we must prove:
\[ P_{\Phi \wedge (i,j')} = P'. \]
We do this by demonstrating an bijection\footnote{This is where our proof differs from Proposition 11 from \cite{TCS2009} where an injection was sufficient.} from the set 
\[\{ 
\pi \in \ppi : \mbox{ $\pi$ is consistent with $\Phi$ and $\pi(i)$ is undefined}\}\] to the set 
\[\{ 
\pi \in \ppi : \mbox{ $\pi$ is consistent with $\Phi \wedge (i,j')$}\}.\]
Given $\pi$ in the former set, with $i$ the unmapped element. Let $i'$ be the element that is mapped to $j'$. Construct $\pi'$ from $\pi$ by substituting $i \mapsto j'$ for $i' \mapsto j'$. The function given by $\pi \mapsto \pi'$ is a bijection, and the result follows.

If $i=n$ or $n$ appears in the left-hand side of an atom $\Phi$ (resp., $n$ occurs in the right-hand side of an atom of $\Phi$) then there must be some $i' \in [n-1] \setminus \{ i\}$ (resp., $j' \in [n-1] \setminus \{1\}$) s.t. $i'$ does not appear in the left-hand side (resp., $j'$ in the right-hand side) of any atom of $\Phi$. It is clear by symmetry that the inequality of the proposition holds iff
\[ P_{\langle i',j' \rangle (\Phi \wedge (i,1))} + P_{\langle i',j' \rangle (\Phi \wedge (i,2))} + \ldots + P_{\langle i',j' \rangle (\Phi \wedge (i,n))} = P_{\langle i',j' \rangle (\Phi)}, \]
and the result follows by the previous argument.
\end{proof}
\end{document}